\documentclass[12pt]{article}
\usepackage{ragged2e}
\RaggedRight
\usepackage[utf8]{inputenc}
\usepackage{setspace}
\usepackage{etoolbox}
\AtBeginEnvironment{table}{\singlespacing}
\usepackage[margin=1.0in]{geometry}
\usepackage{graphicx}
\usepackage{amsmath}
\usepackage{amssymb}
\usepackage{amsfonts}
\usepackage{dsfont}
\usepackage{lineno}
\usepackage{xcolor}
\usepackage{hyperref}
\usepackage{subcaption}
\usepackage{algpseudocode}
\usepackage{algorithm}
\usepackage{amsthm}
\usepackage{rotating}
\usepackage{epsfig}
\usepackage{enumerate}
\usepackage{multirow}
\usepackage{graphics}
\usepackage{epstopdf}
\usepackage{url}

\newcommand{\bm}[1]{{\mbox{\boldmath $#1$}}}
\newcommand{\sN}{\scriptscriptstyle}

\newtheorem{theorem}{Theorem}[section]

\newtheorem{corollary}{Corollary}[section]

\usepackage{csquotes}

\usepackage[round,authoryear]{natbib}
\usepackage{array}
\usepackage{makecell}

\usepackage{titling}
\predate{}
\postdate{}
\usepackage{authblk}
\title{A penalized logistic generalized regression estimator}
\date{} 
\author[1,*]{Grayson W. White} 
\author[2]{Kelly S. McConville} 
\author[3]{Cooper S. Schumacher}
\affil[1]{{\small Department of Mathematics and Statistics, Reed College, Portland, OR, USA}}
\affil[2]{{\small Dominguez Center for Data Science, Bucknell University, Lewisburg, PA, USA}}
\affil[3]{{\small Genospace, Los Angeles, CA, USA}}
\affil[*]{{\small Corresponding author: Grayson W. White, gwhite@reed.edu}}

\begin{document}
\maketitle

\begin{abstract} 
Under a model-assisted framework, a penalized logistic generalized regression estimator is developed to estimate a finite population proportion from complex survey data and auxiliary data.  The proposed estimator controls the impact of unnecessary auxiliary variables through a lasso or ridge penalty.  A central limit theorem is derived for the penalized regression coefficients and for the penalized logistic generalized regression estimator.  Through simulations, it is shown that including the penalty increases the efficiency of the estimator when the true model is sparse. An example using United States Forest Service data demonstrates the applicability of the estimator for surveys with many available auxiliary variables.
\end{abstract}


\newpage

\section{Introduction}\label{intro}

The United States Forest Service Forest Inventory and Analysis Program (FIA) is tasked with monitoring the health and status of the nation's forests.  They produce state and domain estimates of means and totals for numerous forest attributes such as number of trees, biomass, merchantable volume, percent of forest land, mortality, forest type, and more. Under their current standard reporting procedures, FIA estimates forest parameters via post-stratification by combining data collected from a quasi-systematic sample of ground plots with data from high-resolution auxiliary variables collapsed into one categorical variable \citep{bec05}.  Because FIA has access to several high-resolution auxiliary data products, such as satellite imagery and topographic maps, utilizing the availabilty auxiliary data in a multivariable way could increase the efficiency of FIA's estimators. 

A common technique for incorporating auxiliary data is through a model-assisted estimator, such as the generalized regression estimator (GREG) \citep[][Section 6.3]{cas76,sar92}.  The GREG assumes the relationship between the study variable and auxiliary data can be captured with a linear model.  Since the estimator's variance can be reduced by removing extraneous variables, \citet{mcc17} introduced a survey-weighted version of Tibshirani's (\citeyear{tib96}) lasso criterion to estimate the regression coefficients (LASSO).  By penalizing the sum of the absolute value of the coefficient estimates, the LASSO shrinks the estimated coefficients, forcing some to zero and thereby reducing the model size.  The variability of an estimator can also be controlled by penalizing the sum of the squared values of the coefficient estimates.  This penalization results in ridge regression coefficient estimates when minimizing the residual sum of squares.  In the survey sampling setting,
\citet{bar84} and \citet{Chambers1996} considered the ridge penalization for a linear working model under a model-based paradigm. To construct stable calibration weights, the ridge penalization has also been explored in the design-based context \citep{RaoSingh1997, Theberge2000, BeaumontBocci2008}.

In this paper, we consider not a quantitative attribute but a binary, categorical attribute and are interested in estimating a population proportion. Since the logistic model is better suited than the linear model to handle a binary outcome, \citet{leh98} proposed the logistic generalized regression estimator (LGREG).  However, the LGREG also loses efficiency when extraneous variables are included in the model, warranting a lasso or ridge penalty. Such penalties for generalized linear models have been studied \citep{tib96, par07} but in a non-survey sampling context.  Therefore, we propose the penalized logistic generalized regression estimator (PLGREG) to estimate the population proportion.  This estimator utilizes a logistic working model and includes a penalization on the coefficient estimates.

The paper is organized as follows.  Section~\ref{sec:ests} includes the notation for the model-assisted framework.  In Section~\ref{subsec:glmlasso}, we summarize \citet{leh98}'s logistic generalized regression estimator and introduce the lasso and ridge penalized versions in Section~\ref{subsec:llasso}.  Section~\ref{sec:llassoresults} contains a central limit theorem of the penalized estimators. The technical conditions and a central limit theorem for the penalized logistic regression coefficients can be found in Appendix~\ref{sec:appendixA}.  Section~\ref{simulation} reports simulation results which show that the penalized estimators outperform the non-penalized estimators when the true model is sparse and highlights the importance of utilizing a logistic model instead of a linear model when estimating a finite population proportion. The estimators are applied to a FIA estimation problem in Section~\ref{sec:app}.

\section{Penalized logistic generalized regression estimators}\label{sec:ests}

We want to estimate the proportion of the population that are in a particular category of a categorical attribute of interest.  To do so, suppose the units of a finite population, denoted by $U$, can be indexed by $1, 2, \ldots, N$ and let $y_i$ equal 1 if the $i$-th unit is in category $c$ of the attribute of interest and zero otherwise.  Under this notation, the desired population proportion is given by $P_y = N^{-1} \sum_{i \in U} y_i$.  Assume a sample, $s \subset U$, is collected using a probability sampling design $p(\cdot)$, where $p(s)$ is the probability of selecting a particular sample $s$. Let $I_i$ equal 1 if the $i$-th element of $U$ is included in $s$ and 0 otherwise. Similarly, let $I_{ij}$ equal 1 if the $i$-th and $j$-th elements are in $s$ and 0 otherwise. The first and second order inclusion probabilities are $\pi_i = E_p(I_i) = P(i \in s)$ and $\pi_{ij} = E_p(I_i I_j) = P(i, j \in s)$, respectively, where $E_p$ denotes expectation with respect to the design.  Let $\bm{x}_i = (1, x_{1i}, \ldots, x_{di})$ represent the $d$ auxiliary covariates for the $i$-th unit and assume $\{\bm{x}_i\}_{i \in U}$ are known.

\subsection{Logistic generalized regression estimators}\label{subsec:glmlasso}

Under a model-assisted framework, a superpopulation model is used to describe how the finite population values are generated.  Employing a logistic regression model, assume the values of the study variable in the finite population, $\{y_i\}_{i \in \sN{U}}$, are independent realizations from a Bernoulli random variable.  For $i \in U$, let $$P(y_i = 1| \bm{x}_i; \bm{\beta}) = \mu(\bm{x}_i^{\sN{T}} \bm{\beta}) = \exp\left(\bm{x}_i^{\sN{T}} \bm{\beta}\right)\left[1 + \exp\left(\bm{x}_i^{\sN{T}} \bm{\beta}\right)\right]^{-1}$$ with the logit link function
\begin{align}\label{logisticmodel}
\mbox{logit}(\mu(\bm{x}_i^{\sN{T}} \bm{\beta})) = \log \left( \frac{\mu(\bm{x}_i^{\sN{T}} \bm{\beta})}{1 - \mu(\bm{x}_i^{\sN{T}} \bm{\beta})} \right) = \bm{x}_i^{\sN{T}} \bm{\beta}
\end{align}

where $\bm{\beta} = (\beta_o, \beta_1, \ldots, \beta_d)^{\sN{T}}$.  The finite population estimator of $\mu(\bm{x}^{\sN{T}} \bm{\beta})$ is given by $\mu(\bm{x}^{\sN{T}} \bm{\beta}_{\sN{N}})$ where the regression coefficients are estimated by minimizing the negative log-likelihood:
\begin{align*}
 \bm{\beta}_{\sN{N}} &=   \underset{\bm{\beta}}{\arg\min} \left[ -  \sum_{i \in U} \left\{y_i \bm{x}_i^{\sN{T}} \bm{\beta} -  \log \left[1 + \exp(\bm{x}_i^{\sN{T}} \bm{\beta}) \right] \right\} \right].
\end{align*}
Computing $\mu(\bm{x}_i^{\sN{T}} \bm{\beta}_{\sN{N}})$ for each $i \in U$, we can estimate $P_y$ with the difference estimator
\begin{align}
\widehat{P}_{y}^{\sN{DIFF}} = \frac{1}{N}\left[\sum_{i \in s} \frac{y_i - \mu(\bm{x}_i^{\sN{T}} \bm{\beta}_{\sN{N}})}{\pi_i} +  \sum_{i \in U} \mu(\bm{x}_i^{\sN{T}} \bm{\beta}_{\sN{N}}) \right], \label{diffest}
\end{align}
which is design unbiased and has design variance equal to 
\begin{align}\label{diffestvar}
\mbox{Var}_{\rm{p}}\left(\widehat{P}_{y}^{\sN{DIFF}}\right)=\frac{1}{N^2}\underset{i,j \in U}{\sum \sum} \pi_{ij} - \pi_i \pi_j\frac{y_i - \mu(\bm{x}_i^{\sN{T}} \bm{\beta}_{\sN{N}})}{\pi_i}\frac{y_j - \mu(\bm{x}_j^{\sN{T}} \bm{\beta}_{\sN{N}})}{\pi_j}
\end{align}
\citep{sar92}.  To highlight the value of a model-assisted estimator over a purely design-based estimator, notice that the Horvitz-Thompson estimator \citep{hor52} and its design variance are given in equations (\ref{diffest}) and (\ref{diffestvar}) when $\mu(\bm{x}_i^{\sN{T}} \bm{\beta}_{\sN{N}})$ is set to zero for all $i \in U$.  By comparing variances one can see that the  efficiency of the difference estimator over the Horvitz-Thompson estimator depends on how small the residuals, $\{y_i - \mu(\bm{x}_i^{\sN{T}} \bm{\beta}_{\sN{N}})\}_{i \in U}$, are.  

Since $y_i$ is only available for $i \in s$, $\bm{\beta}_{\sN{N}}$ cannot be computed and must be estimated by minimizing the survey-weighted negative log-likelihood, a Horvitz-Thompson estimator of the finite population negative log-likelihood:
\begin{align*}
 \hat{\bm{\beta}}&=   \underset{\bm{\beta}}{\arg\min} \left[ -  \sum_{i \in s} \frac{1}{\pi_i} \left\{y_i \bm{x}_i^{\sN{T}} \bm{\beta} -  \log \left[1 + \exp(\bm{x}_i^{\sN{T}} \bm{\beta}) \right] \right\} \right].
\end{align*}  
In the difference estimator, the finite population mean function, $\mu(\bm{x}^{\sN{T}} \bm{\beta}_{\sN{N}})$, is now estimated by $\mu(\bm{x}^{\sN{T}} \hat{\bm{\beta}})$ which results in 
\begin{align}\label{LGREG}
\hat{P}^{LGREG}_y = \frac{1}{N} \left( \sum_{i \in s} \frac{y_i - \mu(\bm{x}_i^{\sN{T}} \hat{\bm{\beta}}) }{\pi_i} + \sum_{i \in \sN{U}}\mu(\bm{x}_i^{\sN{T}} \hat{\bm{\beta}}) \right),
\end{align}
the logistic generalized regression estimator for the population proportion \citep{leh98}.

\subsection{Penalized logistic generalized regression estimator}\label{subsec:llasso}

If the model in (\ref{logisticmodel}) is sparse, meaning some of the coefficients are exactly zero, then estimating these coefficients can add unnecessary variability to the estimator.   An estimator's variability can also be impacted by highly correlated auxiliary variables.  To remedy these issues, we propose penalizing the size of the coefficient estimates, which will help stabilize the estimates and possibly shrink the model.  For a given penalty parameter $\lambda \geq 0$, the penalized estimates for the coefficients are
\begin{align}\label{coefPLGREG}
 \widehat{\bm{\beta}}^{\gamma} &=    \underset{\bm{\beta}}{\arg\min} \left[ -  \sum_{i \in s} \frac{1}{\pi_i} \left\{y_i \bm{x}_i^{\sN{T}} \bm{\beta} -  \log \left[1 + \exp(\bm{x}_i^{\sN{T}} \bm{\beta}) \right] \right\}  + \lambda\sum_{j=1}^d |\beta_j|^{\gamma} \right]
\end{align}

\noindent where $\hat{\beta}_{o}^{\gamma}$ is unpenalized.  While $\gamma$ could be any positive value, we consider $\gamma = 1$ which corresponds to a lasso penalty and $\gamma = 2$ which corresponds to a ridge penalty.  Both penalties shrink the estimated coefficients towards zero, but only the lasso performs model selection by forcing some estimates to be exactly zero.  The ridge penalty, on the other hand, is good for controlling instability caused by multi-collinearity in the auxiliary variables.

Inserting the penalized coefficient estimates into the estimated mean function gives the penalized logistic generalized regression estimator (PLGREG) of the population proportion,
\begin{align}\label{PLGREG}
\hat{P}^{\sN{PLGREG}}_y = \frac{1}{N} \left( \sum_{i \in s} \frac{y_i - \mu\left(\bm{x}_i^{\sN{T}} \hat{\bm{\beta}}^{\gamma}\right)}{\pi_i} + \sum_{i \in \sN{U}}\mu\left(\bm{x}_i^{\sN{T}} \hat{\bm{\beta}}^{\gamma}\right) \right).
\end{align}
Note, the survey-weighted coefficient vector $
\hat{\bm{\beta}}^{\gamma}$ is a function of the penalty parameter
$\lambda$.  In the simulation study of \S \ref{simulation}, we use
survey-weighted cross validation to select the optimal $\lambda$.  We also compare the performance of the estimator for $\gamma = 1$ and $\gamma = 2$.

To estimate the design-based variance of the PLGREG, we propose the standard variance estimator,
\begin{align}\label{var.est}
\widehat{\text{Var}}(\hat{P}^{\sN{PLGREG}}_y)= \frac{1}{N^2}\displaystyle\sum\limits_{i \in s}\displaystyle\sum\limits_{j \in s}\frac{\pi_{ij} - \pi_i \pi_j}{\pi_{ij}}\frac{\left[y_i - \mu\left(\bm{x}_i^{\sN{T}} \hat{\bm{\beta}}^{\gamma}\right)\right]}{\pi_i}\frac{\left[y_j - \mu\left(\bm{x}_j^{\sN{T}} \hat{\bm{\beta}}^{\gamma}\right)\right]}{\pi_j}
\end{align}
as is suggested for the LGREG by \citet{leh98}.  In the simulations in \S \ref{simulation} and the application in \S \ref{sec:app}, we compute the variance estimator given in Equation~\ref{var.est} for each estimator considered, replacing $\mu\left(\bm{x}_i^{\sN{T}} \hat{\bm{\beta}}^{\gamma}\right)$ with the appropriate estimated mean function. This estimator tends to under-estimate the variance, partly because it does not account for the variability induced by estimating the model coefficients.  Therefore, we also compare the performance of the standard variance estimator to a simple bootstrap variance estimator.  Using the technique described in \S 3 of \citet{mas16}, several with replacement, bootstrap samples are taken from the original sample and the estimate is computed for each bootstrap sample.  The variance estimator is given by the variance of the bootstrap estimates after making a bias adjustment to account for the without replacement design.  More complex bootstrap procedures which handle multi-stage sampling designs and stratification are beyond the scope of this paper.

\section{Asymptotic properties of the penalized logistic generalized regression estimators}\label{sec:llassoresults}

In this section, we present the design-based asymptotic normality for the PLGREG when $\gamma = 1$ (LLASSO) or $\gamma =2$ (LRIDGE).  Asymptotically, the penalized estimator has the same design properties as the LGREG and the variance estimator, given in Equation~\ref{var.est}, is consistent.  Following the standard set-up for survey sampling asymptotics, we consider a sequence of nested populations  $U_1 \subset U_2 \subset \cdots \subset U_{\sN{N}} \subset \cdots$. Let $s_{\sN{N}}\subset U_{\sN{N}}$ be selected according to a sampling design $\rm{p}_{\sN{N}}(\cdot)$ with fixed sample size $n_{\sN{N}}$. We suppress  the subscript $N$ in $n$ as well as in $\pi_j$ and $\pi_{jk}$  for simplicity of notation but utilize the subscript in $\widehat{\bm{\beta}}^{\sN{\gamma}}_{\sN{N}}$ and $\bm{\beta}_{\sN{N}}$ to emphasize that the coefficients change with $N$. The assumptions and a proof of a central limit theorem for the coefficient estimators are given in Appendix~\ref{sec:appendixA}. Before stating the main results, we make a few remarks regarding the assumptions.

\textit{Remark 1.}  By assumptions (\ref{lambda.N}) and (\ref{nrateN}), we allow the penalty parameter to increase as the sample size and population size increase but we keep the number of predictors fixed.  As discussed in \citet{mcc17}, the following properties should continue to hold if the number of predictors increase at a slower rate than the population size so that the sampling error, not the modeling error, remains the dominant source of variability.  In many natural resource inventory applications, the population size, $N$, represents the number of pixels for the landscape and is much larger than the number of auxiliary images, $d$, available for the landscape.  While asymptotic results where $d$ increases as $N$ increases are likely of interest in other survey sampling settings, such formulations are beyond the scope of this paper.

\sloppy\textit{Remark 2.} To prove asymptotic normality of $\widehat{\bm{\beta}}^{\sN{\gamma}}_{\sN{N}}$ requires a Taylor expansion of $\log \left[1 + \exp(\bm{x}_i^{\sN{T}} \bm{\beta}) \right]$, a component of the survey-weighted log-likelihood.  Assumptions (A3) and (A4) provide the regularity conditions needed for the second order term to converge in probability and for the remainder term to be asymptotically negligible.

\textit{Remark 3.} The remaining assumptions, (\ref{HT.cov.big})-(\ref{est.cov.mat}), provide central limit theory for the Horvitz-Thompson estimators needed in our derivations.  For commonly used sampling designs, such as simple random sampling with and without replacement, a central limit theorem for Horvitz-Thompson estimators holds.  See \citeauthor{ful09} (2009, Section 1.3) for more examples of where normal theory holds under a complex design.

 The following results provide asymptotic normality for the penalized logistic generalized regression estimator.  The proofs are omitted as they follow directly from Theorem \ref{thm:CLT.beta.design}.

\begin{theorem}\label{thm:CLT.P.hat.lasso}
Under assumptions (\ref{lambda.N})--(\ref{HT.CLT.big}) and for $\gamma=1,2$, the estimator $\widehat{P}_{\sN{N}}^\gamma$ is asymptotically equivalent to the difference estimator, given in Equation~\ref{diffest},
in the sense that
\begin{align}
\sqrt{n}\left(\widehat{P}_{\sN{N}}^\gamma -\widehat{P}_{y}^{\sN{DIFF}}\right)=o_{\rm{p}}(1), \nonumber
\end{align}
so that
\begin{align*}
\left\{\mbox{Var}_{\rm{p}}\left(\widehat{P}_{y}^{\sN{DIFF}}\right) \right\}^{-1/2} \left(\widehat{P}_{\sN{N}}^\gamma - P_y\right) \overset{D}{\rightarrow} N(0,1).
\end{align*}

\end{theorem}

\begin{theorem}\label{thm:var.est.result}
Under assumptions (\ref{lambda.N})--(\ref{est.cov.mat}) and for $\gamma = 1,2$,

\begin{align}
\widehat{V}(\widehat{P}_{\sN{N}}^\gamma) 
&= \mbox{Var}_{\rm{p}}\left( \widehat{P}_{y}^{\sN{DIFF}} \right) + o_{\rm{p}} \left( n^{-1}\right). \nonumber 
\end{align}
\end{theorem}

Since the variance estimator of the penalized logistic generalized regression estimator is asymptotically consistent for the variance of the difference estimator, the variance estimator can be inserted into the central limit theorem result.  The result provides justification for  normal-based confidence intervals when the sample size is large.

\begin{corollary} \label{cor:CLT.t.hat.lasso}
Under assumptions (\ref{lambda.N})--(\ref{est.cov.mat}) and for $\gamma=1,2$,

\begin{align*}
\left\{\widehat{V}(\widehat{P}_{\sN{N}}^\gamma) \right\}^{-1/2} \left(\widehat{P}_{\sN{N}}^\gamma - P_y\right) \overset{D}{\rightarrow} N(0,1).
\end{align*}
\end{corollary}

\section{Simulation study}\label{simulation} 

To address the comparative performance of the estimators for finite samples, we conduct a simulation study of the following estimators:

\singlespacing
\begin{tabular}{l|l|l|l}\hline
\textbf{Estimator} & \textbf{Working Model} & \textbf{Penalized} & \textbf{Reference}\\
\hline\hline
LLASSO & Logistic & Yes & Equation~\ref{PLGREG} with $\gamma = 1$\\
LRIDGE & Logistic & Yes & Equation~\ref{PLGREG} with $\gamma = 2$\\
LGREG & Logistic & No & Equation~\ref{LGREG} \\
LASSO & Linear & Yes & \citet{mcc17} \\
GREG & Linear & No & \citet{cas76}\\ 
PS & Linear with & No & \citet{coc77}\\
& categorical predictor &&\\
HT & None & No & Equation~\ref{diffest} with $\mu(\bm{x}_i^{\sN{T}} \bm{\beta}_{\sN{N}}) = 0$
\end{tabular}
\\

\doublespacing
\noindent Including different combinations of working models and penalizations allows conclusions to be drawn about when to utilize a particular working model and about how the penalizations impact the efficiency of the estimators under various situations.  The PS is included since this is FIA's standard estimation method.  We examine the bias and efficiency of the estimators when the sparsity of the model and the strength of the model are varied.  We also present the bias of the variance estimators and the confidence interval coverage.

 \subsection{Set-Up} \label{setup}
To generate a single population $U$ of size $N$=10,000, a set of $d$ covariates is drawn from the multivariate normal distribution with a mean vector of zeros and a variance-covariance matrix where the value in the $i$th row and $j$th column is $0.2^{\lvert i-j\rvert}$ for $i, j \in \{1, 2, \ldots, d\}$. The values of the response variable are generated from independent draws of a Bernoulli distribution with probability of success given by the logistic model found in Equation~\ref{logisticmodel} where the values of $\bm{\beta}$ change for the different scenarios considered. To build the PS, a categorical variable is created by transforming the most predictive variable into two categories: $\{x_k \geq 0\}, \{x_k <0\}$.

For each set-up, after generating a single finite population, $\{\bm{x}_j, y_j\}_{j \in U}$, 1,000 samples of size 200 are selected using simple random sampling without replacement.  For each sample, the estimators are computed along with the standard variance estimator given in Equation~\ref{var.est} and the bootstrap variance estimator described in \S \ref{subsec:llasso}.  All estimators are fit using the \texttt{mase} R package \citep{mase}. Within the \texttt{mase} R package, the coefficients for the lasso and ridge estimators are computed using the $\texttt{glmnet}$ function in the R package $\texttt{glmnet}$ \citep{fri10} and the penalty parameter is found via cross-validation using the $\texttt{cv.glmnet}$ function. Note that while $\texttt{glmnet}$ optimizes a slightly different equation than Equation~\ref{coefPLGREG}, the resulting estimated coefficients are equivalent (see Appendix~\ref{sec:appendixB} for details).  The $\texttt{boot}$ function in the R package $\texttt{boot}$ \citep{can16} is used to draw the bootstrap samples.

For each scenario considered, there is a single finite population and therefore the empirical design mean squared error, variance, and bias can be computed for every estimator.  (We found the percent relative design bias of each estimator to be less than 2$\%$ and so do not present the values here.)  The percent relative bias of the variance estimators and the empirical coverage of 95$\%$ confidence intervals for the population proportion are also found.

\subsection{Sparsity trials} \label{sparsity}

\begin{table}[ht]
\centering
\begin{tabular}{rrrrrrr}
  \hline
 & LRIDGE&LGREG & LASSO& GREG & PS & HT \\ 
  \hline
no sparsity (20/20)  & 1.10 & 1.36 & 1.62 & 1.64 & 3.25 & 3.41 \\
  half sparsity (10/20) & 1.11 & 1.21 & 1.47 & 1.50 & 2.81 & 3.08 \\
no model (0/20) & 0.99 & 1.11 & 0.99 & 1.13 & 0.99 & 0.99 \\
   \hline
\end{tabular}
\caption{Ratio of design MSE for each estimator to design MSE of LLASSO while varying the sparsity of the logistic model}
\label{Table:Ratio_MSE_sparsity}
\end{table}

For the sparsity simulations, we set $d$ equal to 20 and generate $y$ under three scenarios: a no sparsity model, a half sparse model, and a completely sparse model.  In each scenario, half of the non-zero coefficients are set to 2, the other half to 0.5 and the intercept $\beta_o$ is adjusted so that the true population proportion is around 0.36.  When no auxiliary variables relate with $y$, implying no model should be used, the HT, PS, and the penalized estimators (LLASSO, LRIDGE, and LASSO) all have roughly the same efficiency and are more efficient than the other unpenalized estimators.  As the sparsity decreases, the HT and PS are much less efficient than the other estimators and the penalized estimators outperform the corresponding unpenalized estimators. The methods based on a logistic model out-perform those based on the linear model but not significantly.  The lasso penalization is more efficient than the ridge penalty when the model is present.  Simulations with models of size 40 were also run but the algorithm to fit the coefficients of the LGREG often failed to converge so the results are not displayed.

\begin{table}[ht]
\centering
\begin{tabular}{rrrrrrrr}
  \hline
& LLASSO & LRIDGE &LGREG & LASSO& GREG &PS & HT \\ 
\hline
no sparsity & -47.96 & -32.22 & -76.08 & -18.19 & -20.01 & -1.93 & -1.43 \\ 
no sparsity (bootstrap)  & 0.42 & -0.40 & 21.07 & 0.99 & 1.92 & -0.92 & -1.46 \\ 

  half sparsity & -32.40 & -26.99 & -56.88 & -11.68 & -16.77 & 0.66 & 0.20 \\ 
  half sparsity (bootstrap)    & -0.96 & 1.03 & 40.96 & 2.19 & 6.60 & 2.02 & 0.07 \\ 

   no model  & -9.18 & -8.61 & -25.10 & -8.64 & -26.01 & -6.87 & -6.31 \\
   no model (bootstrap)   & -8.26 & -7.87 & -10.06 & -7.84 & -5.63 & -5.98 & -6.24 \\ 

   \hline
\end{tabular}
\caption{Percent relative bias of variance estimators \label{table.var.est}}
\end{table}

\begin{table}[ht]
\centering
\begin{tabular}{rrrrrrrr}
  \hline
& LLASSO & LRIDGE &LGREG & LASSO& GREG &PS & HT \\ 
  \hline
no sparsity & 81.20 & 87.50 & 54.40 & 91.80 & 91.50 & 94.40 & 94.90 \\ 
no sparsity (bootstrap)  & 93.90 & 94.70 & 95.90 & 95.20 & 94.90 & 94.50 & 94.50 \\
half sparsity & 87.50 & 90.40 & 75.50 & 92.60 & 92.40 & 94.80 & 94.80 \\ 
  half sparsity (bootstrap) & 94.30 & 94.10 & 97.40 & 94.10 & 95.00 & 95.10 & 94.40 \\
   no model & 92.40 & 92.80 & 91.40 & 92.60 & 90.80 & 92.80 & 92.60 \\ 
   no model (bootstrap) & 92.80 & 92.90 & 93.50 & 93.10 & 94.30 & 92.80 & 93.00 \\   
\hline
\end{tabular}
\caption{Empirical confidence interval coverage for 95$\%$ confidence intervals\label{table.ci}}
\end{table}

 The standard variance estimator of the model-assisted estimators has negative bias, resulting in confidence intervals that are too narrow as observed in Tables \ref{table.var.est} and \ref{table.ci}. The logistic estimators have the worst coverage.  As the sparsity of the model increases, the model-assisted estimators tend to have better confidence interval coverage.  The penalized estimators get closer to the nominal confidence level while the non-penalized methods continue to have under-coverage.  The bootstrap variance estimators have less bias and regardless of the model size, they tend to produce confidence intervals that are fairly close to the nominal rate.  The bootstrap estimator is a good substitution for the standard variance estimator when the design is simple random sampling and a large number of auxiliary variables are available.

\subsection{Strength of model trials}\label{strength}

It is our perception that the GREG is more commonly employed than the LGREG, even when the study variable is categorical.  To study the impact of model misspecification, simulations were run where the strength of the model varied and the level of sparsity was held fixed.  Along with the intercept coefficient, two out of ten coefficients in the model were non-zero, with varying magnitudes.  The non-zero coefficient values can be found in Table \ref{Table:Ratio_MSE_magnitude}. The value of the intercept coefficient was selected so that the proportion of successes in the finite population was roughly 0.36.  Table \ref{Table:Ratio_MSE_magnitude} shows that as the logistic model strengthens, most of the estimators constructed from a logistic working model are more efficient than the other estimators.  The notable exception is LRIDGE which is rather negatively impacted by the large magnitude of the coefficients.  When the model is fairly sparse with a few large effects, the LLASSO is the preferred estimator.

\begin{table}[ht]
\centering
\begin{tabular}{rrrrrrrr}
  \hline
  & LRIDGE &LGREG & LASSO & GREG & PS & HT \\ 
  \hline
$\beta_o=-1.0$, $\beta_9=2.0$, $\beta_{10}=0.5$ & 1.05 & 1.05 & 1.06 & 1.09 & 1.21 & 1.62 \\
 $\beta_o=-1.5$, $\beta_9=4.0$, $\beta_{10}=1.0$  & 1.19 & 1.04 & 1.38 & 1.41 & 1.59 & 2.85 \\ 
  $\beta_o=-3.0$, $\beta_9=8.0$, $\beta_{10}=2.0$  & 1.57 & 1.18 & 1.89 & 1.96 & 2.11 & 4.62 \\ 
 $\beta_o=-6.0$, $\beta_9=16$, $\beta_{10}=4.0$ & 2.36 & 1.43 & 2.98 & 3.08 & 3.33 & 7.53 \\ 
   \hline
\end{tabular}
\caption{Ratio of design MSE for each estimator to design MSE of LLASSO while varying the coefficients in the logistic model}\label{Table:Ratio_MSE_magnitude}
\end{table}

\section{Application to the Forest Inventory and Analysis Program}\label{sec:app}

\begin{figure}[h!]
  \centering
    \includegraphics[width=0.75\textwidth]{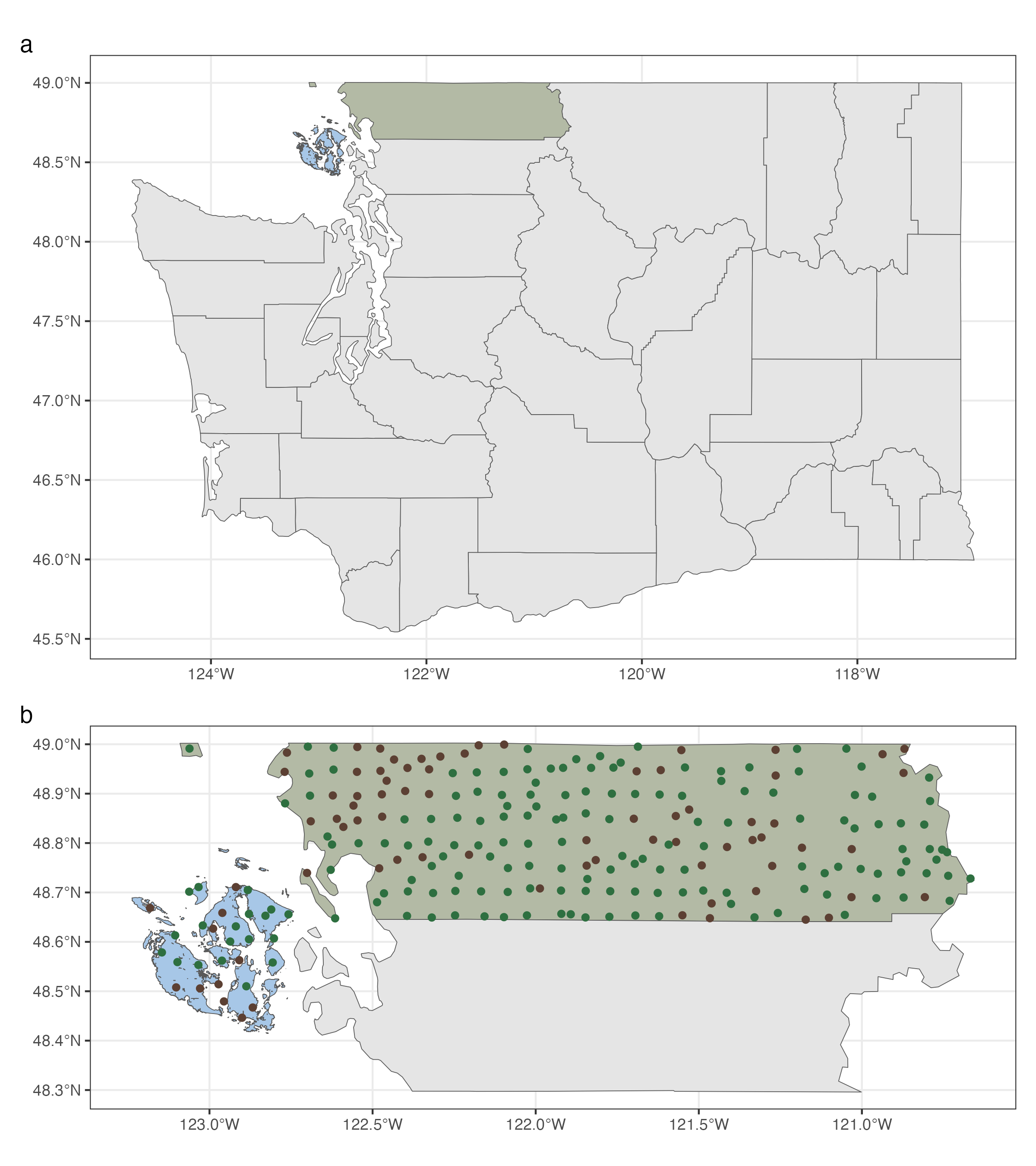} 
 \caption{\label{fig:region} Study regions of interest (a) and sampled points in the study regions of interest (b). In both subfigures, San Juan County is depicted in blue and Whatcom County is depicted in green. In subfigure (b), each point corresponds to an individual sample plot location, with point colors corresponding to forested (green) and unforested (brown) plots.}
 \end{figure}

We estimated the proportion of forested land, as defined in \citet{bec05}, in two counties in northwest Washington State: San Juan County and Whatcom County.  Since FIA produces estimates over spatial domains of the US, the study areas are discretized where each point or pixel of $U$, represents an equally sized section of land.  The two study areas, San Juan and Whatcom County, are colored in blue and green respectively in Figure \ref{fig:region}.  San Juan County is an archipelago of 176 (named) islands situated in the Salish Sea in the northwest corner of Washington State, with islands that have a range of forest cover proportions: from almost none, to nearly fully forested. Whatcom county is a nearly contiguous county in northwest Washington, however a few noncontiguous regions are part of Whatcom county; islands (most notably Lummi Island) and a pene-exclave (Point Roberts). San Juan County is much smaller than Whatcom County: it is discretized into 55,015 90 by 90 meters sqaured pixels and only contains 30 sampled plot locations. On the other hand, Whatcom County is 692,900 pixels and contains 212 sampled plot locations. Sampled plots were collected by FIA via an equal-probability quasi-systematic sampling design, defined in \citet{bec05}. 

FIA sample data containing our response variable (an indicator for if a plot is forested) along with auxiliary variables was sourced from the \texttt{forested} R data package \citep{forested}. Population auxiliary data was prepared by the FIA Program and then clipped and extracted for our two domains of interest. Fifteen auxiliary data products are available for every pixel in the two counties. Descriptions of the one response variable and fifteen auxiliary variables are included in Table~\ref{tab:vardescription}. More information on the auxiliary variables can be found in \citet{forested}. 

       \begin{table}[h!] 
       \centering

\resizebox{5in}{!}
{%
  \begin{tabular}{c|l}
\hline
\textbf{Code Name} & \textbf{Description} \\ 
  \hline
  elevation & Elevation, in meters. \\
  \hline
  eastness & Transformed aspect degrees to eastness (-100 to 100). \\
  \hline
  northness & Transformed aspect degrees to northness (-100 to 100). \\
  \hline
  roughness & Degree of irregularity of the plot. \\
  \hline
  tree\_no\_tree & LANDFIRE tree/non-tree lifeform mask \\
  \hline
  dew\_temp & Mean annual dewpoint temperature (1991-2020), in degrees Celsius. \\
  \hline
  precip\_annual & Mean annual precipitation (1991-2020), in mm × 100. \\
  \hline
  temp\_annual\_mean & Mean annual temperature (1991-2020), in degrees Celsius. \\
  \hline
  temp\_annual\_min & Mean annual minimum temperature (1991-2020), in degrees Celsius. \\
  \hline
  temp\_annual\_max & Mean annual maximum temperature (1991-2020), in degrees Celsius. \\
  \hline
  temp\_january\_min & Mean minimum temperature in January (1991-2020), in degrees Celsius. \\
  \hline
  vapor\_min & Minimum annual vapor pressure deficit (1991-2020), in Pa x 100. \\
  \hline
  vapor\_max & Maximum annual vapor pressure deficit (1991-2020), in Pa x 100. \\
  \hline
  canopy\_cover & Analytical Tree Canopy Cover, as a percent. \\
  \hline
  land\_type & Land cover type from European Space Agency (ESA) 2020 \\ & WorldCover global land cover product \\
  \hline
  \end{tabular} }
  \caption{Variable descriptions of response and auxiliary variables.\label{tab:vardescription}}
  \end{table}

We computed the following estimators: LLASSO, LRIDGE, LGREG, LASSO, GREG, PS, and HT. As in the simulation study, all estimators are computed via the \texttt{mase} R package \citep{mase}. To compute PS in our analyses, we use the land\_type variable. For variance estimation, the quasi-systematic sampling design is approximated by a simple random without replacement sampling design as suggested by \citet{bec05}.  

Table \ref{tab:vars} provides the selected variables for the lasso methods where the penalty parameter was chosen through 10-fold cross-validation. We see that the lasso method based on a linear model keeps a large number of auxiliary variables, even in San Juan County where there are only 30 sampled plots, where the logistic lasso keeps only four and six variables in San Juan and Whatcom counties, respectively. The only variables retained for every lasso method were eastness, precip\_annual, vapor\_max, and canopy\_cover.

 \begin{table}[h]
\centering
\begin{tabular}{|c|r|r||r|r|}
  \hline
&\multicolumn{2}{c||}{\textbf{San Juan}}&\multicolumn{2}{c|}{\textbf{Whatcom}}\\ 
& \small{LLASSO} & \small{LASSO} & \small{LLASSO} & \small{LASSO}\\
\hline
elevation & & X  & X & X\\
eastness &X& X &X &X\\
northness && X &X &X\\
roughness & & X & &X\\
tree\_no\_tree &  &  & &X\\
dew\_temp & & & &X\\
precip\_annual  &X &X &X &X\\
temp\_annual\_mean &  &  &&\\
temp\_annual\_min &  & X &&\\
temp\_annual\_max &  &  &&X\\
temp\_january\_min &  & X &&X\\
vapor\_min &  &  &&X\\
vappor\_max & X & X &X&X\\
canopy\_cover & X & X &X&X\\
land\_type &  & X &&X \\
   \hline
\end{tabular}
\caption{\label{tab:vars}Variables selected by the lasso methods.}
\end{table}

In San Juan County, where the sample size is limited to only 30 observations, the estimators performed quite differently to each other. Notably, the LLASSO and LRIDGE estimators produced the smallest estimated bootstrap standard errors and their estimates aligned well with the PS. The LGREG, LASSO, and GREG seem to have performed sporadically: their closed-form and bootstrap standard errors did not align, likely due to the over-parametrization of the model. This is most noticeable with the GREG estimator, where we see a closed-form standard error of 0.039 and a bootstrap standard error of 33.481. In the larger sample size Whatcom County, estimators performed more similarly to each other. Point estimates ranged from 0.656 to 0.693, and standard errors ranged from 0.020 to 0.032. All estimators had smaller estimated standard errors than the HT, but the difference in estimated standard errors between the model-assisted estimators was negligible. Across counties, the bootstrapped standard errors are equal to or larger than the standard errors found using the square root of Equation~\ref{var.est}, especially as the size of the model increases, except for PS and HT in Whatcom County where they are smaller by 0.001. Since Equation~\ref{var.est} is known to have negative bias, as observed in \S \ref{sparsity},  we used the bootstrap standard error when comparing the efficiency of the estimators. For San Juan County, the penalized regression estimators are more efficient than the PS and HT, but the unpenalized regression estimators are less efficient. In Whatcom County, HT is the least efficient estimator, and PS is just as efficient and sometimes more efficient than the estimators which incorporate the other variables.   

  \begin{table}[h]
\centering
\begin{tabular}{|r|r|r|r||r|r|r|}
  \hline
&\multicolumn{3}{c||}{\textbf{San Juan}}&\multicolumn{3}{c|}{\textbf{Whatcom}}\\ 
&&& \multicolumn{1}{c||}{\small{Bootstrap}} &&&\multicolumn{1}{c|}{\small{Bootstrap}} \\
& \multicolumn{1}{c|}{\small{Point}} & \multicolumn{1}{c|}{\small{Standard}} &\multicolumn{1}{c||}{\small{Standard}}  &\multicolumn{1}{c|}{\small{Point}} & \multicolumn{1}{c|}{\small{Standard}} & \multicolumn{1}{c|}{\small{Standard}} \\
& \multicolumn{1}{c|}{\small{Estimate}} & \multicolumn{1}{c|}{\small{Error}} &\multicolumn{1}{c||}{\small{Error}}  &\multicolumn{1}{c|}{\small{Estimate}} & \multicolumn{1}{c|}{\small{Error}} & \multicolumn{1}{c|}{\small{Error}} \\
\hline
  LLASSO & 0.707 & 0.058 & 0.069 & 0.672 & 0.023 & 0.024 \\ 
  LRIDGE & 0.661 & 0.071 & 0.076 & 0.670 & 0.023 & 0.024 \\ 
  LGREG & 0.502 & 0.000 & 0.104 & 0.666 & 0.020 & 0.024 \\ 
  LASSO & 0.556 & 0.059 & 0.119 & 0.662 & 0.021 & 0.022 \\ 
  GREG & 0.401 & 0.039 & 33.481 & 0.656 & 0.021 & 0.024 \\ 
PS &  0.758 & 0.080 &  0.091 & 0.658 & 0.024 & 0.023 \\
HT & 0.633 & 0.089 & 0.094 & 0.693 & 0.032 & 0.031 \\ 
   \hline
\end{tabular}
\caption{\label{tree.table}Point estimates of the proportion of forested land in two counties of Washington State and the corresponding estimated standard errors.}
\end{table}

\section{Conclusion}\label{sec:conclusion}

We introduced a penalized logistic generalized regression estimator for complex survey data. Under standard assumptions for design-based asymptotics, we show that the introduced estimator has asymptotic normality and that the standard model-assisted variance estimator is design-consistent. Through a simulation study, we showed the estimator performs well for finite samples and studied its performance as we vary model sparsity and strength. Finally, in a data application, we showcase the practical use of the estimator for the FIA Program. 

The FIA Program, like many other organizations who report official statistics, has access to many sources of auxiliary data and these data have the potential to increase the efficiency of estimation methods.  However, including unnecessary data can negatively impact that efficiency.  Penalized estimators constrain the coefficient estimates and can lead to a more stable solution.  In turn, the regression estimators built on these penalized techniques, like the one introduced here, tend to be more efficient than the regression estimators based on the full model, especially if the true model is sparse.  For their official reporting, the FIA Program currently uses post-stratification to estimate population quantities, regardless of whether or not the forest attribute is quantitative or categorical.  As the simulation study and application show, there are opportunities for the FIA Program to gain efficiency over their standard post-stratification based estimation by including important variables not used in the creation of the post-strata. The penalized techniques presented in this paper can help the FIA Program find these relevant variables.

\clearpage

\section*{Acknowledgements}
The authors would like to thank the US Forest Service, Forest Inventory and Analysis Program for the data.

\section*{Data availability statement}
The sample data used in our data application can be accessed in the \texttt{forested} R package. However, pixel-level population auxiliary data were sourced from confidential raster data, which cannot be shared publicly. Requests for data used here or other requests including confidential data should be directed to FIA’s Spatial Data Services (\url{https://www.fs.usda.gov/research/programs/fia/sds}).

\clearpage
\bibliography{logisticlassopaper}

\clearpage

\appendix
\renewcommand{\thefigure}{A.\arabic{figure}}  
\renewcommand{\thetable}{A.\arabic{table}}   
\setcounter{figure}{0}                      
\setcounter{table}{0}                         

\section{Appendix}\label{sec:appendixA}

\small

To prove Theorem \ref{thm:CLT.P.hat.lasso}, we first need to find the asymptotic distribution of the survey-weighted, penalized coefficient estimators. Assume the following conditions as $N \rightarrow \infty$ with $d$ fixed:

\renewcommand{\theenumi}{A\arabic{enumi}}
\begin{enumerate}

\item \label{lambda.N} The penalty parameter satisfies $\lambda_{\sN{N}} = o(\sqrt{N})$.

\item The sampling rate $nN^{-1} \rightarrow \pi \in (0,1)$.   \label{nrateN}

\item  \label{V1.mat}
Let $\phi(x) = \log(1+e^{x})$. The matrix
\begin{align*}
\widehat{\bm{V}}_{\sN{N}1} = \frac{1}{N} \sum_{i \in U_{\sN{N}}} \bm{x}_i  \bm{x}_i^{\sN{T}} \phi^{''}(\bm{x}_i^{\sN{T}} \bm{\beta}_{\sN{N}}) \frac{I_i}{\pi_i}
\end{align*}
is positive definite and there exists positive definite matrix $\bm{V}_1$ such that

$\widehat{\bm{V}}_{\sN{N}1} - \bm{V}_1 = o_{\rm{p}}(1)$ elementwise.

\item \label{phi3} For all $\bm{\beta} \in Q$, a sufficiently large open set that contains $\bm{\beta}_{\sN{N}}$, assume
\begin{align*}
(p+3)(p+2)(p+1) \frac{1}{N}\sum_{i \in U_{\sN{N}}} \phi^{'''}(\bm{x}_i^{\sN{T}} \bm{\beta}_{\sN{N}}) \max_{1 \leq j,k,l \leq p} |x_j x_k x_l| \frac{I_i}{\pi_i} = O_{\rm{p}}(1).
\end{align*}

\item \label{HT.cov.big} The matrix
$\bm{\Sigma} = \underset{N \rightarrow \infty}{\lim}  \bm{\Sigma}_{\sN{N}}$ exists and is positive definite, where $\bm{\Sigma}_{\sN{N}}$ is the design covariance matrix
\begin{align*}
\bm{\Sigma}_{\sN{N}} &= \begin{bmatrix} \bm{\Sigma}_{\sN{N}}^{(yxyx)} & \bm{\Sigma}_{\sN{N}}^{(yx\phi^{'}x)} \\
\bm{\Sigma}_{\sN{N}}^{(\phi^{'}xyx)} & \bm{\Sigma}_{\sN{N}}^{(\phi^{'}x\phi^{'}x)}
\end{bmatrix} \\
&= \begin{bmatrix} \frac{n}{N^2}\underset{i,j \in U_{\sN{N}}}{\sum \sum} \frac{\pi_{ij} - \pi_i \pi_j}{\pi_i \pi_j} y_i \bm{x}_i y_j\bm{x}_j^{\sN{T}} & \frac{n}{N^2}\underset{i,j \in U_{\sN{N}}}{\sum \sum} \frac{\pi_{ij} - \pi_i \pi_j}{\pi_i \pi_j} \phi^{'}(\bm{x}_i^{\sN{T}} \bm{\beta}_{\sN{N}}) \bm{x}_i y_j \bm{x}_j^{\sN{T}} \\
\frac{n}{N^2}\underset{i,j \in U_{\sN{N}}}{\sum \sum} \frac{\pi_{ij} - \pi_i \pi_j}{\pi_i \pi_j} \phi^{'}(\bm{x}_i^{\sN{T}} \bm{\beta}_{\sN{N}}) \bm{x}_i y_j \bm{x}_j^{\sN{T}}&  \frac{n}{N^2}\underset{i,j \in U_{\sN{N}}}{\sum \sum} \frac{\pi_{ij} - \pi_i \pi_j}{\pi_i \pi_j} \phi^{'}(\bm{x}_i^{\sN{T}} \bm{\beta}_{\sN{N}}) \bm{x}_i \phi^{'}(\bm{x}_j^{\sN{T}} \bm{\beta}_{\sN{N}}) \bm{x}_j^{\sN{T}}\end{bmatrix}
\end{align*}
of the following $2(p+1)$ vector of centered, standardized Horvitz-Thompson estimators:
\begin{align}\label{z.big}
\bm{z}_{\sN{N}} = \begin{bmatrix} \frac{\sqrt{n}}{N} \sum_{i \in U_{\sN{N}}} y_i\bm{x}_i \left( \frac{I_i}{\pi_i} -1 \right) \\ \frac{\sqrt{n}}{N} \sum_{i \in U_{\sN{N}}} \phi^{'}(\bm{x}_i^{\sN{T}} \bm{\beta}_{\sN{N}}) \bm{x}_i \left( \frac{I_i}{\pi_i} -1 \right) \end{bmatrix}.
\end{align}

\item \label{HT.CLT.big} The normalized, centered, Horvitz-Thompson estimators defined in (\ref{z.big}) satisfy a central limit theorem:
$\bm{z}_{\sN{N}} \overset{D}{\rightarrow} \mathcal{N} (\bm{0}, \bm{\Sigma})$.

\item \label{est.cov.mat} The subvector $\bm{z}_{\sN{N}}^* = (z_{\sN{N}1}, z_{\sN{N}p+2})$  of the vector $\bm{z}_{\sN{N}}$ defined in (\ref{z.big}) has a design-consistent covariance matrix estimator
\begin{align*}
\widehat{\bm{\Sigma}}_{\sN{N}}^* & =
\begin{bmatrix} \frac{n}{N^2}\underset{i,j \in U_{\sN{N}}}{\sum \sum} \frac{\pi_{ij} - \pi_i \pi_j}{\pi_i \pi_j} \frac{I_i I_j}{\pi_{ij}} y_i y_j& \frac{n}{N^2}\underset{i,j \in U_{\sN{N}}}{\sum \sum} \frac{\pi_{ij} - \pi_i \pi_j}{\pi_i \pi_j} \frac{I_i I_j}{\pi_{ij}} y_i \phi^{'}(\bm{x}_j^{\sN{T}} \hat{\bm{\beta}}_{\sN{N}}) \\
\frac{n}{N^2}\underset{i,j \in U_{\sN{N}}}{\sum \sum} \frac{\pi_{ij} - \pi_i \pi_j}{\pi_i \pi_j} \frac{I_i I_j}{\pi_{ij}} \phi^{'}(\bm{x}_i^{\sN{T}}\hat{\bm{\beta}}_{\sN{N}}) y_j& \frac{n}{N^2}\underset{i,j \in U_{\sN{N}}}{\sum \sum} \frac{\pi_{ij} - \pi_i \pi_j}{\pi_i \pi_j} \frac{I_i I_j}{\pi_{ij}} \phi^{'}(\bm{x}_i^{\sN{T}}\hat{\bm{\beta}}_{\sN{N}})\phi^{'}(\bm{x}_j^{\sN{T}} \hat{\bm{\beta}}_{\sN{N}})\end{bmatrix}
\\&
=\begin{bmatrix}  \widehat{\Sigma}_{\sN{N}}^{(yy)}&\widehat{\Sigma}_{\sN{N}}^{(y\phi^{'})} \\ \widehat{\Sigma}_{\sN{N}}^{(\phi^{'}y)} & \widehat{\Sigma}_{\sN{N}}^{(\phi^{'}\phi^{'})}\end{bmatrix},
\end{align*}
in the sense that $\widehat{\bm{\Sigma}}_{\sN{N}}^* - \bm{\Sigma}_{\sN{N}}^* = o_{\rm{p}}(1)$ elementwise
where $\bm{\Sigma}_{\sN{N}}^*$ is the covariance matrix of $\bm{z}_{\sN{N}}^*$.

\end{enumerate}

 \setcounter{theorem}{0}
  \renewcommand{\thetheorem}{A.\arabic{theorem}}

\begin{theorem}\label{thm:CLT.beta.design}
Under assumptions (\ref{lambda.N})--(\ref{HT.CLT.big}) and for $\gamma =1,2$, the penalized coefficient estimates $\widehat{\bm{\beta}}^{\sN{\gamma}}_{\sN{N}}$ satisfy
\begin{align*}
\sqrt{N} \left(\widehat{\bm{\beta}}^{\sN{\gamma}}_{\sN{N}} - \bm{\beta}_{\sN{N}} \right) \overset{D}{\rightarrow} \mathcal{N} \left(\bm{0}, \pi^{-1}\bm{V}_1^{-1} \bm{V}_2 \bm{V}_1^{-1}\right)
\end{align*}
as $N\to\infty$, where the matrix $\bm{V_2}$ is defined by
\begin{align*}
\bm{V_2} =  \bm{\Sigma}^{(yxyx)}- 2\bm{\Sigma}^{(yx\phi^{'}x)} + \bm{\Sigma}^{(\phi^{'}x\phi^{'}x)}.
\end{align*}

\end{theorem}

Similar to Theorem 2.1 in \citet*{mcc17}, we will prove Theorem \ref{thm:CLT.beta.design} by finding the asymptotic distribution of the criterion which the estimator minimizes. 

\begin{proof}[Proof of Theorem \ref{thm:CLT.beta.design}]

For $\bm{u} \in \mathbb{R}^{p+1}$, write the sample criterion based on the survey-weighted negative log-likelihood as 
\begin{align*}
\bm{A}_{\sN{N}}(\bm{u}) :=   \sum_{i \in s} \frac{1}{\pi_i} \left\{-y_i \bm{x}_i^{\sN{T}} \bm{u} +  \phi\left(\bm{x}_i^{\sN{T}} \bm{u} \right) \right\} + \lambda_{\sN{N}} \sum_{j=1}^d |u_j|^{\gamma}
\end{align*}
where $\widehat{\bm{\beta}}_{\sN{N}} = \underset{\bm{u}}{\arg\min} \bm{A}_{\sN{N}}(\bm{u})$ and let
\begin{align*}
\bm{C}_{\sN{N}}(\bm{u}) := \bm{A}_{\sN{N}}\left(\bm{\beta}_{\sN{N}} + \bm{u}\frac{\sqrt{n}}{N}\right) - \bm{A}_{\sN{N}}(\bm{\beta}_{\sN{N}}).
\end{align*}
Then $N n^{-1/2} ( \widehat{\bm{\beta}}_{\sN{N}} - \bm{\beta}_{\sN{N}})= \underset{\bm{u}}{\arg\min} \bm{C}_{\sN{N}}(\bm{u})$. 
Notice that $\phi^{'}(x) = \mu(x)$.  Applying Taylor's Theorem to $\phi(\cdot)$ and utilizing the fact that 
\begin{align*}
\frac{\sqrt{n}}{N} \sum_{i \in U_{\sN{N}}} \left[y_i - \phi^{'}\left(\bm{x}_i^{\sN{T}}\bm{\beta}_{\sN{N}}\right) \right]\bm{x}_i^{\sN{T}} =0,
\end{align*}
 we get
\begin{align*}
\bm{C}_{\sN{N}}(\bm{u}) =& 
-\frac{\sqrt{n}}{N} \sum_{i \in s} \frac{y_i \bm{x}_i^{\sN{T}}}{\pi_i} \bm{u} +\sum_{i \in s} \left\{ \phi\left(\bm{x}_i^{\sN{T}}\left[\bm{\beta}_{\sN{N}} + \bm{u}\frac{\sqrt{n}}{N}\right]\right) -  \phi\left(\bm{x}_i^{\sN{T}}\bm{\beta}_{\sN{N}}\right)  \right\}\frac{1}{\pi_i} \\
 &+ \lambda_{\sN{N}} \sum_{j=1}^d\left( \left|\beta_{\sN{N}j} + \frac{u_j\sqrt{n}}{N} \right|^{\gamma} - \left|\beta_{\sN{N}j} \right|^{\gamma} \right) \\
&= 
-\frac{\sqrt{n}}{N} \sum_{i \in s} \frac{1}{\pi_i}\left[y_i - \phi^{'}\left(\bm{x}_i^{\sN{T}}\bm{\beta}_{\sN{N}}\right) \right]\bm{x}_i^{\sN{T}}\left(\frac{I_i}{\pi_i} -1 \right)\bm{u} +\bm{u}^{\sN{T}}\frac{n}{2N^2} \sum_{i \in s} \frac{1}{\pi_i} \phi^{''}\left(\bm{x}_i^{\sN{T}}\bm{\beta}_{\sN{N}}\right)  \bm{x}_i\bm{x}_i^{\sN{T}}\bm{u} \\
 &+6^{-1} \left(\frac{\sqrt{n}}{N}\right)^{3} \sum_{i \in s} \frac{1}{\pi_i} \phi^{'''}\left(\bm{x}_i^{\sN{T}}\bm{\beta}^{o}\right)  \left(\bm{x}_i^{\sN{T}}\bm{u}\right)^{3}  + \lambda_{\sN{N}} \sum_{j=1}^d\left( \left|\beta_{\sN{N}j} + \frac{u_j\sqrt{n}}{N} \right|^{\gamma} - \left|\beta_{\sN{N}j} \right|^{\gamma} \right)\\
& := C_{\sN{N}1}(\bm{u}) + C_{\sN{N}2}(\bm{u})  + C_{\sN{N}3}(\bm{u}) + C_{\sN{N}4}(\bm{u}) 
\end{align*}
where $\bm{\beta}^{o}$ is between $\bm{\beta}_{\sN{N}}$ and $\bm{\beta}_{\sN{N}} + \bm{u}\sqrt{n}N^{-1}$. Assumptions (\ref{HT.cov.big}) and (\ref{HT.CLT.big}) imply
\begin{align*}
 C_{\sN{N}1}(\bm{u}) \overset{D}{\rightarrow} \mathcal{N} \left(\bm{0}, \bm{V}_2\right) \bm{u}^{\sN{T}}
\end{align*}
since the variance of the centered Horvitz-Thompson estimator can be written
\begin{align*}
\mbox{Var}_{\rm{p}}\left(-\frac{\sqrt{n}}{N} \sum_{i \in s} \frac{1}{\pi_i}\left[y_i - \phi^{'}\left(\bm{x}_i^{\sN{T}}\bm{\beta}_{\sN{N}}\right) \right]\bm{x}_i^{\sN{T}}   \right) 
&=   \bm{\Sigma}^{(yxyx)}_{\sN{N}}- 2\bm{\Sigma}^{(yx\phi^{'}x)}_{\sN{N}} + \bm{\Sigma}^{(\phi^{'}x\phi^{'}x)}_{\sN{N}}\\
&:= \bm{V}_{\sN{N}2}
\end{align*}
and by assumption (\ref{HT.cov.big}),  $\bm{V} = \underset{N \rightarrow \infty}{\lim}  \bm{V}_{\sN{N}}$.  The second term, $C_{\sN{N}2}(\bm{u}) $, converges in probability to $2^{-1}\bm{u}^{\sN{T}}  \pi \bm{V}_1 \bm{u}$ by assumptions (\ref{nrateN}) and  (\ref{V1.mat}). By assumption (\ref{phi3}), $C_{\sN{N}3}(\bm{u}) = O_p(N^{-1/2})$.  For the last term,
$$
C_{\sN{N}4}(\bm{u})  = \frac{\sqrt{n}}{N} \lambda_{\sN{N}} \sum_{j=1}^d \mbox{sign}(\beta_{\sN{N}j}) |u_j|^{\gamma} \leq \frac{\sqrt{n}}{N} \lambda_{\sN{N}} \sum_{j=1}^d |u_j|^{\gamma} =\frac{\sqrt{n}}{N} \lambda_{\sN{N}} ||\bm{u}||_{l\gamma}^{\gamma} =O(N^{-1/2})
$$
by assumptions (\ref{lambda.N}) and (\ref{nrateN}).  By Slutsky's Theorem, we get 
\begin{align*}
\bm{C}_{\sN{N}}(\bm{u}) \overset{D}{\rightarrow} \bm{C}(\bm{u}) 
\end{align*}
where
\begin{align*}
\bm{C}(\bm{u}) = - \bm{u}^{\sN{T}} N (\bm{0}, \bm{V_2}) + \frac{1}{2} \bm{u}^{\sN{T}}  \pi \bm{V}_1 \bm{u}.
\end{align*}
Since $\bm{C}_{\sN{N}}(\bm{u})$ is convex and $\bm{C}(\bm{u})$ has a unique minimum, we have the result.  

\end{proof}

\clearpage

\section{Appendix}\label{sec:appendixB}

\small

We'd like to ensure that the estimate of our regression coefficients, $\widehat{\bm{\beta}}^{\gamma}$, in Equation~\ref{coefPLGREG} is properly estimated by \texttt{glmnet}. Recall Equation~\ref{coefPLGREG}:
\begin{align*}
 \widehat{\bm{\beta}}^{\gamma} &=    \underset{\bm{\beta}}{\arg\min} \left[ -  \sum_{i \in s} \frac{1}{\pi_i} \left\{y_i \bm{x}_i^{\sN{T}} \bm{\beta} -  \log \left[1 + \exp(\bm{x}_i^{\sN{T}} \bm{\beta}) \right] \right\}  + \lambda\sum_{j=1}^d |\beta_j|^{\gamma} \right].
\end{align*}
When supplying \texttt{glmnet} with survey weights $w_i = \frac{1}{\pi_i}$, \texttt{glmnet} normalizes these weights before optimization by replacing the supplied weights with $w_i\left({\sum_{i\in s}w_i}\right)^{-1}$. Further, \texttt{glmnet} optimizes a different equation for lasso ($\gamma = 1$) and ridge ($\gamma = 2$) regression. In particular, \texttt{glmnet} optimizes:

\begin{align*}
 \widehat{\bm{\beta}}^{\gamma}_{\texttt{glmnet}} &=   \begin{cases}
  \underset{\bm{\beta}}{\arg\min} \left[ -  \sum_{i \in s} \frac{w_i}{W} \left\{y_i \bm{x}_i^{\sN{T}} \bm{\beta} -  \log \left[1 + \exp(\bm{x}_i^{\sN{T}} \bm{\beta}) \right] \right\}  + \lambda\sum_{j=1}^d |\beta_j|^{1} \right] & \text{for } \gamma = 1 \\
  \underset{\bm{\beta}}{\arg\min} \left[ -  \sum_{i \in s} \frac{w_i}{W} \left\{y_i \bm{x}_i^{\sN{T}} \bm{\beta} -  \log \left[1 + \exp(\bm{x}_i^{\sN{T}} \bm{\beta}) \right] \right\}  + \frac{\lambda}{2}\sum_{j=1}^d |\beta_j|^{2} \right] & \text{for } \gamma = 2
\end{cases} 
\end{align*}

where $W = \sum_{i\in s}w_i$. We will show $\widehat{\bm{\beta}}^{\gamma} = \widehat{\bm{\beta}}^{\gamma}_{\texttt{glmnet}}$.

\begin{proof}

Let $\lambda_{\texttt{glmnet}} = \lambda  W^{-1}$ if $\gamma = 1$, and $\lambda_{\texttt{glmnet}} = 2\lambda  W^{-1}$ if $\gamma = 2$. Then for lasso regression:

\begin{align*}
\widehat{\bm{\beta}}^{\gamma = 1}_{\texttt{glmnet}} &=  \underset{\bm{\beta}}{\arg\min} \left[ -  \sum_{i \in s} \frac{w_i}{W} \left\{y_i \bm{x}_i^{\sN{T}} \bm{\beta} -  \log \left[1 + \exp(\bm{x}_i^{\sN{T}} \bm{\beta}) \right] \right\}  + \lambda_{\texttt{glmnet}}\sum_{j=1}^d |\beta_j|^{1} \right] \\
&=  \underset{\bm{\beta}}{\arg\min} \left[  \frac{1}{W} \left(- \sum_{i \in s} w_i \left\{y_i \bm{x}_i^{\sN{T}} \bm{\beta} -  \log \left[1 + \exp(\bm{x}_i^{\sN{T}} \bm{\beta}) \right] \right\}  + W\lambda_{\texttt{glmnet}}\sum_{j=1}^d |\beta_j|^{1} \right) \right] \\
&= \underset{\bm{\beta}}{\arg\min} \left[ - \sum_{i \in s} w_i \left\{y_i \bm{x}_i^{\sN{T}} \bm{\beta} -  \log \left[1 + \exp(\bm{x}_i^{\sN{T}} \bm{\beta}) \right] \right\}  + W\lambda_{\texttt{glmnet}}\sum_{j=1}^d |\beta_j|^{1}  \right] \\
&= \underset{\bm{\beta}}{\arg\min} \left[ - \sum_{i \in s} w_i \left\{y_i \bm{x}_i^{\sN{T}} \bm{\beta} -  \log \left[1 + \exp(\bm{x}_i^{\sN{T}} \bm{\beta}) \right] \right\}  + \lambda \sum_{j=1}^d |\beta_j|^{1}  \right] \\
&= \widehat{\bm{\beta}}^{\gamma = 1}.
\end{align*}

Similar logic applies to ridge regression.  \qedhere

\end{proof}

\end{document}